\documentclass[letterpaper,twocolumn,10pt]{article}
\usepackage{usenix,epsfig}

\usepackage{amsmath}
\usepackage{color}
\usepackage{hyperref}
\usepackage{algorithm}
\usepackage[noend]{algpseudocode}
\usepackage{subcaption}

\usepackage{pgfplots}
\pgfplotsset{compat=1.12}
\usetikzlibrary{backgrounds}
\usepackage{tikzscale}

\newcommand{\link}[2]{\href{#1}{\color{blue}#2}}

\newcommand{\hide}[1]{{}}
\newcommand{\lineref}[1]{{Line~\ref{#1}}}
\newtheorem{remark}{Remark}
\newtheorem{claim}{Claim}
\newtheorem{proof}{Proof}

\makeatletter
\renewcommand{\ALG@beginalgorithmic}{\small}
\makeatother

\begin{document}

\date{}

\title{\Large \bf Bizur: A Key-value Consensus Algorithm for Scalable File-systems}

\author{
{\rm Ezra N. Hoch}\\
ezra.hoch@elastifile.com
\and
{\rm Yaniv Ben-Yehuda}\\
yaniv.benyehuda@elastifile.com
\and
{\rm Noam Lewis}\\
noam.lewis@elastifile.com
\and
{\rm Avi Vigder}\\
avi.vigder@elastifile.com
} 

\maketitle


\subsection*{Abstract}
Bizur is a consensus algorithm exposing a key-value interface.
It is used by a distributed file-system that scales to 100s of servers, delivering
millions of IOPS, both data and metadata, with consistent low-latency.

Bizur is aimed for services that require strongly consistent state, but do not require
a distributed log; for example, a distributed lock manager or a distributed service locator.
By avoiding a distributed log scheme, Bizur outperforms
distributed log based consensus algorithms, producing more IOPS and
guaranteeing lower latencies during normal operation and especially
during failures.

Paxos-like algorithms (e.g., Zab and Raft) which are used by existing distributed file-systems,
can have artificial contention points due to their dependence on a distributed log.
The distributed log is needed when replicating a general service, but when the desired
service is key-value based, the contention points created by the distributed log
can be avoided.

Bizur does exactly that, by reaching consensus independently on independent keys.
This independence allows Bizur to handle failures more efficiently and to scale much better than
other consensus algorithms, allowing the file-system that utilizes Bizur to scale with it.

\section{Introduction}
File-systems are a very convenient way to access data. As systems and data centers become larger,
their requirements from file-systems scale as well. The crux of scaling a file-system is to
scale while providing strong consistency.

File-system operations can be divided into two
main categories: data operations, such as \textit{read} and \textit{write}, and metadata operations,
like \textit{create file} and \textit{lookup}. Usually, it is easier to scale the data operations,
while scaling the metadata operations is harder. The main reason being that preserving metadata
consistency is more subtle than preserving data's consistency. For example, two writes to different
files cannot affect each other, while two \textit{rename} operations on different directories could
affect each other - since they might create a loop in the directory structure.

Many distributed file-systems utilize Paxos as their configuration service or even as their
metadata service. Examples of such file-systems include Ceph \cite{weil2006ceph,cephusingpaxos},
GFS (Google Filesystem) via its Chubby service \cite{burrows2006chubby},
XtreemFS \cite{hupfeld2008xtreemfs,xtreemfsusingpaxos}, Infinit \cite{infinituseingpaxos}
and Nutanix \cite{nutanixbible}.

In this paper we present Bizur, a new key-value consensus algorithm. Bizur is used as
an alternative for Paxos,
both as a configuration service and especially as the underlying infrastructure
for a scalable metadata service. Bizur's main properties are: high-concurrency, linear scalability and
low latency while preserving strict consistency.

\subsection{Motivation}\label{subsec:motivation}
At \link{http://www.elastifile.com}{Elastifile} we've built a distributed scale-out file-system,
that is designed to scale to 100s and 1000s of servers. To achieve that, the file-system's metadata services
are truly distributed: each file or directory can be owned by a different server,
and can migrate as the workload changes.
This fine-grained dynamic metadata ownership allows to scale linearly both the data and the metadata operations.

Since files and directories are dynamically owned, there is a distributed
repository that keeps track of which server is
the owner of which file / dir. This service is called the ownership repository coordinator, a.k.a. ORC.
The ORC service must provide a (strongly) consistent view of its state, since the consistency of the
entire file-system depends on the ORC.

Many user-initiated
requests, especially metadata requests, could potentially query the ORC service to know which server is handling the relevant file / dir.
Therefore, the concurrency
of operations against the ORC service could be very high. In addition, our file-system is required to provide consistent performance,
with the lowest possible degradation during failure and recovery.

Previous consensus algorithms like Paxos \cite{lamport2001paxos}, Zab \cite{junqueira2011zab} or Raft \cite{ongaro2014search} don't handle these requirements well enough
(see Subsection~\ref{subsection:distlog}). Bizur was created to provide the high-concurrency
and the consistent low-latency that is required by Elastifile's services. Subsection~\ref{subsection:bizurway}
gives an overview on how Bizur achieves that.

Bizur is also used for Elastifile's Cluster Database (ECDB); this service provides a persistent,
consistent and fault-tolerant view of the cluster's state (sometimes called ``configuration service'').
The main difference between the usage of Bizur by the ORC service and by the ECDB service is that
the ORC service is in-memory, while the ECDB is persisted. This allows the ORC service
to be blazingly fast, without harming the consistency of the service.
Regular failures are handled by Bizur's fault tolerance, and total power failure
is handled by reconstructing the ORC's state from information persisted during transactions\footnote{Exact specification of the ORC service is outside the scope of this paper.}.

\begin{remark}
The results presented in
Section~\ref{sec:experimental} refer to the persisted version of Bizur (the ECDB service), so that
they can be compared fairly with the other persisted services: etcd and ZooKeeper. The latencies of
the ORC service are much faster, as it doesn't need to write to disk.
\end{remark}

\subsection{Distributed Log Drawbacks}\label{subsection:distlog}
Consensus algorithms like Paxos, Zab and Raft
are all based on a distributed log. Operations are written to the log, and the state of the system is
advanced by applying the log entries one after the other. Therefore, to reach the state relevant to
some log entry {\it n}, all the log entries preceding it must be applied first.
This creates two kinds of dependencies between entires: ``false'' dependencies between a log
entry and its preceding ones, and ``real'' dependencies between a log entry affecting some content $x$
and preceding entries that affect the same content $x$.

Real dependencies have the following drawback: to be able to read some content $x$, we must
first find all the log entries relating to $x$, and apply them in order. Thus, if some
server has been disconnected for some time, and is now reconnected, it must first recover multiple
log entries before it can access the content of $x$. There is no clear bound on how many such entries
the server must recover. This drawback manifests itself mainly when there are partial disconnections,
and the service continues to advance while some of the servers are not updated.

False dependencies between log entries have two additional different drawbacks. First, a single slow operation
will increase the latency of all succeeding operations, until the slow operation is committed.
For example, a network packet drop will
affect multiple ongoing operations (instead of affecting just the operation within the dropped packet).
Second, during recovery, the distributed log must be recovered and replayed, which can take a
long time (the time it takes to recover and replay the log depends on additional choices, like the rate at which checkpoint is performed).
Both of these drawbacks hamper the ability to deliver consistent latency for user-initiated
IOs.

There is an additional drawback to the high cost of recovery: longer timeouts for failure detection.
To see why, consider the following: on the one hand we want short timeouts to detect real failures
quickly. On the other hand, if the timeouts are too quick we pay the cost of recovery when we
didn't really need to, since it wasn't a real failure (for example, a short network failure,
or a loaded server that didn't get running time for the consensus process).
Thus, there is a trade-off between the time to overcome a real failure and the number of falsely-detected failures. When the cost of
false detection is high, systems have longer timeouts to detect real failures, and thus take longer to return
to normal operation.

Moreover, replicating state using a distributed log means that each replica has some ``base state''
on top of which the log is replayed. This base-state must be updated periodically, to avoid maintaining
an ever growing log. Thus, distributed log based algorithms have an additional flow of log-compaction
(also call checkpointing or snapshotting). During this process the servers agree on some new base-state,
and can erase the log leading up to that state. Log compaction is complex with non-trivial trade-offs
\cite{raftlogcompaction, chandra2007paxos}, and requiring it is another drawback.

To sum up, the drawbacks that consensus algorithms based on a distributed log have are:
a) reading an object requires having all related preceding entries available,
b) a slow operation will affect unrelated succeeding operations,
c) time to detect a failure is longer,
d) once a failure is detected, recovery can take a long time, and
e) it requires handling the complex flow of log-compaction.

\subsection{Our Contribution}\label{subsection:bizurway}
Bizur was created to overcome the aforementioned drawbacks. Bizur is a consensus algorithm that
ensures consistent low latency of operations. It allows for independent operations to run concurrently
without affecting each other (no false dependencies); and it guarantees a constant latency addition
in case of failures (or leader change).
This allows us to have a very low timeout of failure detection; in our implementation, the timeout is 100ms.

Bizur achieves these features by avoiding the usage of a distributed log. Since there is no log,
there is no need for one operation to wait for an unrelated operation to be written to the log,
and no need to wait for log recovery upon failure. As a bonus, there is also
no need to handle log compaction, which can be quite arduous (see \cite{raftlogcompaction, chandra2007paxos}).

The main difference between Bizur and the other consensus algorithms (e.g., Paxos),
is that the distributed log based consensus algorithms solve a more general problem. They are
a generalized solution, that enables consensus for any data model. Bizur, on the other hand, is
optimized for a specific use-case: a strongly-consistent distributed key-value store. Since
Bizur is utilized for a specific use-case, it makes stronger assumptions about the data model,
and can avoid unneeded contention points.

In many cases (i.e., etcd \cite{etcd} or ZooKeeper \cite{hunt2010zookeeper}) key-value
stores are build \textbf{on top of} a distributed log based consensus algorithms (Raft in the case of etcd,
and Zab in the case of ZooKeeper). Those systems pay for the price of a generalized data model,
while they in fact support a much simpler one. Bizur merges the lower-level consensus algorithm with
the higher-level key-value data model. The resulting key-value consensus has concurrency and low-latency
properties that don't exist in systems that are based on the more generalized consensus.

Bizur exploits the simplified key-value data model, by allowing operations on different keys
to advance independently. It can be thought of as if each key has its own distributed log consensus
algorithm. However, that would be \textbf{very} inefficient. Moreover, once each key is updated independently,
we don't really need a log anymore. Instead, Bizur uses shared-memory like
constructs (see \cite{attiya1995sharing}) to reach consensus on the value of each key. Thus
avoiding the need for a state's history of each key-value pair,
and hence not needing a distributed log at all.

The straightforward mapping of keys to multiple ``small'' consensuses (e.g., a consensus instance per key)
has its difficulties (see Subsection~\ref{sec:replication}). Bizur introduces a more subtle mapping,
by hashing the keys into predefined buckets, on which consensus is achieved.

An additional gain of Bizur's key-value data model is that scaling and sharding can be done at the Bizur
level, and not at the application level. When working with distributed log based consensus algorithm
sharding is an application task; the application needs to decide according to what it performs the sharding.
When scaling, the application will need to handle creating new instances of the consensus algorithm to handle the
different shards. This means that each application needs to ``reinvent the wheel'' regarding scaling and sharding.
Due to Bizur's simplified
data model, sharding is trivial, and is done by hashing the key. Thus, scaling and shard management are implemented
once in the Bizur algorithm, and can be reused multiple times by various applications.

\section{The Bizur Algorithm}
The Bizur algorithm is designed to run in real-world environments. More formally, it assumes
asynchronous message passing and supports message drops as well as crash-stop failures.
It does not handle Byzantine failures.

The Bizur algorithm exposes a key-value interface, containing the \textit{get}, \textit{set}, \textit{delete}
and \textit{iterate keys} operations. Each key is mapped into a bucket, and the bucket is replicated
across the cluster. The buckets are independent of each other, allowing for high-concurrency of
operations. Operations on keys that map into the same bucket are serialized.
The number of buckets depends on the number of keys supported by the system, and on the desired
density of keys-to-buckets.

Bizur achieves consensus on each bucket in a similar way to an atomic register \cite{attiya1995sharing}.
To ensure atomicity, Bizur has a single leader per bucket.
For simplicity, we use the same leader for all buckets. We continue the discussion assuming a
single leader; it is easy to extend Bizur to have different leaders for different buckets.


A Bizur's run is split into phases called \textit{elections}. For each election there is at
most one leader. There can be elections with no leader. The leader is responsible for receiving requests,
processing them and replicating the affected bucket.

\begin{remark}\label{remark:brevity}
The algorithmic description assumes the following (the exact details are removed for brevity):
\begin{itemize}
    \item When the leader sends a message to the cluster and awaits a majority
    of responses, it is assumed there is some identifier that associates the responses to the original
    message.
    \item When the leader sends to ``all'', it also sends to itself (as a replica), and this message
          is assumed to be received.
    \item Timeouts are handled as if the leader received a \textsc{Nack} message from the relevant server.

    \item Operations on the same bucket are serialized (e.g., protected by a {\it mutex}).
\end{itemize}
\end{remark}

\begin{algorithm}[t]
\caption{Leader Election}\label{leaderelection}
\begin{algorithmic}[1]
\Procedure{StartElection}{}
\State $\textit{elect\_id} \gets \textit{elect\_id}+1$
\State \textbf{send} \textsc{PleaseVote}(\textit{elect\_id}, \textit{self}) to all
\If {received \textsc{AckVote} from majority}
\State \label{election:true} $\textit{is\_leader} \gets \textit{true}$
\EndIf
\EndProcedure

\Procedure{PleaseVote}{\textit{elect\_id}, \textit{source}}
\If {\textit{elect\_id} $>$ \textit{voted\_elect\_id}}
\State $\textit{voted\_elect\_id} \gets \textit{elect\_id}$
\State $\textit{leader} \gets \textit{source}$
\State \textbf{reply} \textsc{AckVote} to \textit{source}
\ElsIf {\textit{elect\_id} = \textit{voted\_elect\_id} \textbf{and}
\State \text{ ~ ~ } \textit{source} = \textit{leader}}
\State \textbf{reply} \textsc{AckVote} to \textit{source}
\Else ~\textbf{reply} \textsc{NackVote} to \textit{source}
\EndIf
\EndProcedure
\end{algorithmic}
\end{algorithm}

\subsection{Leader Election}\label{subsec:leaderelection}
Bizur requires at most one leader in each election. To achieve that, each server $p$ keeps track
of two variables: \textit{elect\_id} which is the highest election that $p$ tried to get elected in,
and \textit{voted\_elect\_id} which is the highest election that $p$ voted in.

A new election is initiated at server $p$ by calling \textsc{StartElection}, which increments \textit{elect\_id},
and requests the cluster for votes. If enough servers replied with \textsc{AckVote}, then $p$ marks
itself as leader.

A request for a vote is done by sending \textsc{PleaseVote}. When server $q$ receives \textsc{PleaseVote}
from $p$, it checks if it is a new election and if so $q$ records
its vote for $p$; otherwise, $q$ returns \textsc{NackVote} to $p$.

Algorithm~\ref{leaderelection} contains
the pseudo code for the leader election flow.

\begin{claim}\label{claim:leaderelection_safety}
For every \textit{elect\_id} there is at most one server that marks \textit{is\_leader}
as \textit{true} (Line~\ref{election:true}).
\end{claim}
\begin{proof}
    Sketch:
    For every \textit{elect\_id}, every server $q$ votes for at most one server $p$. A server can
    execute Line~\ref{election:true} only if it received \textsc{AckVote} from a majority of the servers.
    If two servers $p, p'$ receive \textsc{AckVote} from a majority of the servers, then there is a server
    $q$ that sent \textsc{AckVote} to $p$ and $p'$. Since $q$ votes for at most one server $p$ (for the
    given \textit{elect\_id}), we
    conclude that $p=p'$.
\end{proof}

Claim~\ref{claim:leaderelection_safety} shows that Algorithm~\ref{leaderelection} is \textit{safe}.
Its \textit{liveness} depends on the actual failure pattern; from our experience, it works
well in practice.

\subsection{Replication}\label{sec:replication}
Bizur's replication is based on the concept of a SWMR register (see \cite{attiya1995sharing} for more details).
Our implementation processes all reads at the leader server, slightly simplifying the flow (since usually there
is a single reader, instead of multiple readers).

Replication is done per bucket (and not per key-value pair), to handle deleted entries. If replication
is done on a key-value pair basis, it would mean that a new key would require a new instance of the
SWMR register, and deleting a key would require relinquishing the SWMR instance. Suppose there are
multiple concurrent sets and deletes on the same key: how can we distinguish between the old and new
instances of the SWMR register?

Bizur solves this problem by defining a priori a fixed number of long-lived buckets, each of which has a
SWMR instance. Each key-value pair is hashed into a bucket, and the updated bucket is then replicated.
Since there are no creations and/or deletions of SWMR instances, the aforementioned problem doesn't exist.

Algorithm~\ref{writebucket} describes the write flow of buckets, and Algorithm~\ref{readbucket}
describes the read flow of buckets. Following is a quick overview of these algorithms. The
\textsc{Write} and \textsc{Read} operations are internal to the Bizur servers, and are not exposed
to the Bizur client.

Buckets are identified by their $index$ (recall there is a fixed number of buckets).
Each bucket has a version $ver$, which is a tuple composed of: a) $elect\_id$ which is the leader's
$elect\_id$ when the bucket was written, and b) $counter$, which increases by one each time the
bucket is written. When there is a new $elect\_id$ the $counter$ is zeroed.

When writing a bucket (see \textsc{Write}) the leader sets the bucket's version and sends the
bucket to all the servers in the Bizur cluster. If the leader receives enough acks, it considers the
write successful; otherwise, it relinquishes its leadership, since it does not have a majority of servers
that think its the leader.

Each replica, when receiving a write request, compares the bucket's
$ver.elect\_id$ to the replica's $voted\_elect\_id$. If the bucket is more up-to-date
the replica updates its ``voting history'' and updates its local view of the bucket.
Since Bizur assumes non-Byzantine failures, the leader
was indeed voted in by a majority of servers, thus updating a server's vote to the majority's vote
won't affect the \textit{safety} of the leader election. If the bucket's version is older than
the replica's latest vote, the replica rejects the write; notifying the originator leader that it
is no longer the leader.

\begin{algorithm}[t]
\caption{Bucket Replication: Write}\label{writebucket}
\begin{algorithmic}[1]
\Procedure{Write}{\textit{bucket}}\Comment{By leader}
\State $\textit{bucket.ver.elect\_id} \gets \textit{elect\_id}$
\State $\textit{bucket.ver.counter} \gets \textit{bucket.ver.counter}+1$
\State \textbf{send} \textsc{ReplicaWrite}(\textit{bucket}, \textit{self}) to all
\If {received \textsc{AckWrite} from majority}
\State \textbf{return} \textit{true}
\Else
\State $\textit{is\_leader} \gets \textit{false}$
\State \textbf{return} \textit{false}
\EndIf
\EndProcedure

\Procedure{ReplicaWrite}{\textit{bucket}, \textit{source}}
\If {\textit{bucket.ver.elect\_id} $<$ \textit{voted\_elect\_id}}
\State \textbf{reply} \textsc{NackWrite} to \textit{source}
\Else
\State $\textit{voted\_elect\_id} \gets \textit{bucket.ver.elect\_id}$
\State $\textit{leader} \gets \textit{source}$ \Comment{``update'' vote}
\State $local\_buckets[bucket.index] \gets bucket$
\State \textbf{reply} \textsc{AckWrite} to \textit{source}
\EndIf
\EndProcedure
\end{algorithmic}
\end{algorithm}

Reading a bucket is slightly more involved, since it also recovers the bucket following
a leader change. First thing \textsc{Read} does is to ensure that the bucket was recovered
(see Subsection~\ref{subsection:recovery}). If the
bucket was recovered, it means the leader has successfully written its view of the bucket previously.
Thus, the leader can return the content of its local bucket (\lineref{line:return_local_content}).
All that is left is to ensure that the leader is indeed the leader, which is accomplished by
sending \textsc{ReplicaRead} to all replicas (\lineref{line:leader_still_leader_1}) and requiring
to receive a majority of acks (\lineref{line:leader_still_leader_2}).

Each replica, when receiving a read request, does the same ``voting history update'' as it does in
the write case. If the replica hasn't voted in a newer election, it will respond with the bucket
from its local view.

\hide{
\begin{remark}\label{remark:lockperbucket}
For simplicity, we disallow multiple writes on the same bucket concurrently, and therefore serialize
writes to the same bucket / key.
\end{remark}
}

\begin{algorithm}[t]
\caption{Bucket Replication: Read}\label{readbucket}
\begin{algorithmic}[1]
\Procedure{Read}{\textit{index}}\Comment{By leader}
\If {not \textsc{EnsureRecovery}(\textit{index}, \textit{elect\_id})}
\State \textbf{return} $\perp$
\EndIf
\State \label{line:leader_still_leader_1} \textbf{send} \textsc{ReplicaRead}(\textit{index}, \textit{elect\_id}, \textit{self}) to all
\If {received \textsc{AckRead} from majority} \label{line:leader_still_leader_2}
\State \label{line:return_local_content} \textbf{return} $local\_buckets[index]$
\Else
\State $\textit{is\_leader} \gets \textit{false}$
\State \textbf{return} $\perp$
\EndIf
\EndProcedure

\Procedure{ReplicaRead}{\textit{index}, \textit{elect\_id},  \textit{source}}
\If {\textit{elect\_id} $<$ \textit{voted\_elect\_id}}
\State \textbf{reply} \textsc{NackRead} to \textit{source}
\Else
\State $\textit{voted\_elect\_id} \gets \textit{elect\_id}$
\State $\textit{leader} \gets \textit{source}$ \Comment{``update'' vote}
\State \textbf{reply} \textsc{AckRead}($local\_buckets[index]$) to \textit{source}
\EndIf
\EndProcedure
\end{algorithmic}
\end{algorithm}

\begin{algorithm}[t]
\caption{Bucket Replication: Recovery}\label{ensurerecovery}
\begin{algorithmic}[1]
\Procedure{EnsureRecovery}{\textit{index}, \textit{elect\_id}}
\If {$elect\_id = local\_buckets[index].ver.elect\_id$}\label{line:verify_election}
\State \textbf{return} \textit{true}
\EndIf
\State \label{line:start_recovery} \textbf{send} \textsc{ReplicaRead}(\textit{index}, \textit{elect\_id}, \textit{self}) to all
\If {received \textsc{AckRead} from majority}
\State \label{line:select_bucket_1} $\textit{max\_ver} \gets \max\{bucket.ver ~|~ \text{received \textit{bucket}}\}$
\State \label{line:select_bucket_2} $bucket \gets \text{some } bucket \text{ s.t. } bucket.ver = max\_ver$
\State \label{line:write_bucket_start} $bucket.ver.elect\_id \gets elect\_id$
\State $bucket.ver.counter \gets 0$
\State \label{line:end_recovery}\textbf{return} \textsc{Write}($bucket$)
\Else
\State $\textit{is\_leader} \gets \textit{false}$
\State \textbf{return} \textit{false}
\EndIf
\EndProcedure
\end{algorithmic}
\end{algorithm}

\subsection{Recovery}\label{subsection:recovery}
Recovery occurs when the leader is replaced, either due to a failure, or because it was disconnected
for a long time. If a server that is not the leader fails (or is disconnected) no recovery occurs.
Recovery is on a per-bucket basis, and can thus occur concurrently and independently on multiple
buckets. Due to this property, Bizur performs the recovery lazily on the first request
(following the leader change) to a bucket.

\textsc{EnsureRecovery} (see Algorithm~\ref{ensurerecovery}) guarantees that a \textsc{Read} that
follows a leader change will first recover the bucket. \lineref{line:verify_election} checks
if the bucket has already been recovered, and if so returns immediately, shortcutting the somewhat
expensive distributed operation that follows. Notice that the local bucket's $elect\_id$ is updated
only during \textsc{Write}, and thus comparing it to the current $elect\_id$ (\lineref{line:verify_election})
does indeed tell us if the bucket is already recovered or not.

If the current $elect\_id$ is newer, then there was a leader change, and recovery is
performed (\lineref{line:start_recovery}-\lineref{line:end_recovery}). Recovery consists of
reading the bucket from a majority of replicas, selecting the bucket with the highest version
(\lineref{line:select_bucket_1}-\lineref{line:select_bucket_2}) and writing it back to the replicas
with an updated version (\lineref{line:write_bucket_start}-\lineref{line:end_recovery}).

Notice that following a leader change, the first read of the bucket will update the
bucket's version $elect\_id$. This ensures that the leader's bucket's version is the
highest among all versions at the replicas (for that specific bucket), since bucket's versions
are compared by first comparing the $elect\_id$ part, and since the leader has a higher
$elect\_id$ from all previous leaders. This flow is crucial for the correctness of Bizur, as
it ensures a value read by the leader won't change without an additional write by the leader. Without
this flow,
there can be ``hidden'' writes that later are seen by the leader, causing the value of the bucket
to change without a write from the leader.

\hide{
\subsection{Clients}\label{sec:clients}
Bizur clients play an active role in the Bizur algorithm. Clients are responsible for sending
requests to the leader, and initiating an election of a new leader if there is none.

Both of
these responsibilities are implemented in a single flow: the client remembers the last valid leader,
and tries sending requests to it. If there is a timeout, or if the client receives a \textsc{NotALeader}
message from that server, the client moves (in a round-robin manner) to the next server.
Once all servers have been tried, the client will request a new election on one of
the servers. By limiting the rate at which a client can initiate a new election, we ensure that
elections can finish successfully, before a new election is started.

For the above flow to work well, all server code paths can return \textsc{NotALeader} to the client. This
can happen if the leader relinquishes its leadership during an operation, or if it isn't a leader
when it received the client's request. For brevity, these paths are not present in the algorithmic
description.
}

\subsection{Key-Value}\label{subsection:keyvalue}
The Bizur clients expose key-value API: \textsc{Get, Set, Delete}
and \textsc{IterateKeys}. When a Bizur client receives a request, it forwards it to the leader.
Algorithm~\ref{keyvalue} described the flow of the operations at the leader.

As described in Section~\ref{sec:replication}, Bizur's internal replication uses buckets.
Algorithm~\ref{keyvalue} shows how the mapping between key-value to buckets occurs.

We assume a bucket encodes key-value pairs in some format.
The following helper functions are used to access this encoding:
$\textbf{hash}(key)$ returns the bucket's index which the $key$ hashes into,
$\textbf{decode}(bucket, key)$ returns the value referred by $key$ within the $bucket$,
$\textbf{encode\_set}(bucket, key, value)$ adds the $(key,value)$ pair to the $bucket$,
$\textbf{encode\_delete}(bucket, key)$ removes the $key$ (and referred value) from the $bucket$,
and $\textbf{decode\_keys}(bucket)$ returns the set of keys encoded in the $bucket$.

Notice that each of the key-value operations performs an internal \textsc{Read} operation. As
seen in Section~\ref{sec:replication}, the \textsc{Read} operation performs ``recovery'' for the
first read occurring after a leader change. Following operations will usually avoid the additional
cost incurred by the recovery (the additional write phase).

Our implementation supports conditional \textsc{Set} and \textsc{Delete}. That is, the mutating
operations can receive an expected value, and perform the set / delete only if the expected value
is indeed the one existing in the bucket. Since all operations go through the leader \hide{, and due to Remark~\ref{remark:lockperbucket}} it is straightforward to add this variant.
For brevity, the details are removed from Algorithm~\ref{keyvalue}.

\begin{algorithm}[t]
\caption{Key-Value API}\label{keyvalue}
\begin{algorithmic}[1]
\Procedure{Get}{\textit{key}}
\State $index \gets hash(key)$
\State $bucket \gets \textsc{Read}(index)$
\State \textbf{return} $decode(bucket, key)$
\EndProcedure

\Procedure{Set}{\textit{key}, \textit{value}}
\State $index \gets hash(key)$
\State $bucket \gets \textsc{Read}(index)$
\State $encode\_set(bucket, key, value)$
\State \textsc{Write}($bucket$)
\EndProcedure

\Procedure{Delete}{\textit{key}}
\State $index \gets hash(key)$
\State $bucket \gets \textsc{Read}(index)$
\State $encode\_delete(bucket, key)$
\State \textsc{Write}($bucket$)
\EndProcedure

\Procedure{IterateKeys}{}
\State $res \gets \emptyset$
\ForAll {$index$}
\State $bucket \gets \textsc{Read}(index)$
\State $res = res \cup decode\_keys(bucket)$
\EndFor
\State \textbf{return} $res$
\EndProcedure
\end{algorithmic}
\end{algorithm}

\subsection{Correctness}
Following is a sketch of the correctness proof of the Bizur algorithm.
Subsection~\ref{subsec:leaderelection} treats the leader election flow, and ensures that
there is at most one leader for every election.
Since each bucket is independent of the others, we can concentrate our analysis on a single $bucket$.

Due to the leader election properties, and due to \textsc{ReplicaWrite} and \textsc{ReplicaRead},
at most one leader is able to work with a majority of the servers. Thus, we can talk about ``the leader''
of $bucket$.

All operations go through the leader of $bucket$, so as long as there is no new election, the leader
can serialize the operations, ensuring linearizability of the bucket. The first operation a new leader
performs ``fixes'' the content of $bucket$, ensuring all future operations of the same leader are valid.
Recall that the leader performs one operation at a time for $bucket$ (Remark~\ref{remark:brevity}).
We're left with ensuring that a leader's change and the first operation of a the new leader maintains
linearizability.

Consider the first operation a new leader does:
\textsc{Write}: The leader will overwrite the previous content
    of $bucket$, thus ignoring the previous content. Notice that the $bucket.ver$ contains the
    new $elect\_id$ and is higher than any previous bucket version $bukcet$ had. This is crucial
    for the \textsc{Read} operation.

\textsc{Read}: The leader will first recover the previous content of $bucket$
(\lineref{line:start_recovery}-\lineref{line:end_recovery} in Algorithm~\ref{ensurerecovery}).
    As part of this recovery, the leader reads the most up to date previous value, updates the
    bucket version and writes $bucket$ back to
    the cluster. Ensuring two properties: a) the content returned by \textsc{Read} is written to a
    majority of the servers, so future reads will see it, and b) the content returned is a content
    written by some previous leader, ensuring consistency.

\begin{remark}
    Notice that until the new leader's first operation, the value of $bucket$ is potentially ``undetermined'',
    since the old leader might be in the middle of a \textsc{Write}. If the new leader starts with a \textsc{Read},
    the content that it will see depends on what replicas received the old leader's messages. Once the
    new leader recovers the bucket, the value is determined, and the old leader can't affect its content anymore.
\end{remark}

To sum up, following a leader's election, the first operation it performs on a given bucket will ``fix''
the content of that bucket. All future operations of the same leader are serialized. Therefore, all
operations on the bucket preserve linearizability, as required.

\begin{algorithm}[t]
\caption{Reconfiguration Read}\label{reconfigread}
\begin{algorithmic}[1]
\Procedure{ReconfigRead}{\textit{index}}\Comment{By leader}
\State $bucket \gets \textsc{Read}(index)$ \Comment{from $new$}
\If {$bucket.needs\_copy$}
\State $bucket \gets \text{\textsc{Read}($index$) from }old$ Bizur instance
\State $bucket.needs\_copy \gets false$
\State $\textsc{Write}(bucket)$ \Comment{to $new$}
\EndIf
\State return $bucket$
\EndProcedure
\end{algorithmic}
\end{algorithm}

\subsection{Reconfiguration}\label{sec:reconfig}
Bizur's support for reconfiguration (change in the cluster's members) was inspired by \cite{Lorch:2006:SWM:1218063.1217946}. The main idea is to have two Bizur instances running $old$ and $new$
and to transfer responsibility between them in a consistent and fault-tolerant fashion.

Consider a Bizur cluster $old$ containing a set of servers $O$, and consider a reconfiguration step
that would like to update the Bizur cluster to be the set of servers $N$. The following supports
$O$ and $N$ being any set of servers, they can be disjoint, intersect or even be the same set (\textit{i.e.,} $O=N$).

First step is to create a new Bizur instance $new$ running on the cluster $N$. The servers in $O \cap N$
will participate in both $old$ and $new$ instances. The $new$ instance is in a \textit{reconfig} state,
which means any request it receives will first copy the bucket from the $old$ instance, then handled by the $new$
instance. The second step notifies the $old$ instance to return a \textsc{ReconfigError} to all requests
(including ongoing ones).
When a client receives \textsc{ReconfigError} it will resend the message to the $new$ instance.
At this stage, we're guaranteed that no client requests will be processed by the $old$ instance.

The third step is to notify all clients that they should access the $new$ instance. This step isn't
always needed, because clients that will contact the $old$ instance will get an error telling them to
contact the $new$ instance. However, eventually the $old$ instance will be removed, and by then all
clients should be aware of the $new$ instance.

The last two steps are cleanup steps: first we wait until the $new$ instance finishes copying all
data from the $old$ instance (this can happen slowly in the background), after which its state changes
from \textit{reconfig} to \textit{normal}. Lastly, the $old$ instance
can be removed.

During the reconfiguration, there are two Bizur instances that are alive, each with its own fault-tolerance.
Therefore, both instances continue operating correctly, even in the presence of failures, making the
reconfiguration process fault-tolerant as well.

When the $new$ instance is in the \textit{reconfig} state, it keeps track of additional information
per bucket, stating if the bucket has already been copied or not. When the $new$ instance was created
all buckets were set to require copy from $old$. When a first access to a bucket occurs, the $new$
instance reads the bucket, sees that it requires copy from the $old$ instance, copies is, and writes
it to the $new$ instance, marking that it doesn't require copying. Note that the \textit{reconfig}
state ends only after all buckets have been copied.

Algorithm~\ref{reconfigread} describes the \textsc{ReconfigRead} method, which is used
instead of the regular \textsc{Read} method when in \textit{reconfig} state.

\hide{
\subsection{Multi-key Transactions}
Multi-key transactions allow to change multiple key-value pairs atomically.
They are an efficient tool in many use cases. In our
implementation we do not need multi-key transactions, since each key-value pair
represents ownership on a single file / directory, and we do not need to change ownership
on multiple files / directories in a transactional way.

Multi-key transactions can be implemented in Bizur as an additional (higher) layer,
by utilizing the leader to run a transactional protocol (somewhat similar to two-phase
commit) over the key-value pairs. The ability to perform conditional write
simplifies the transactional flow.

The details of the multi-key transaction implementation are outside the scope of this paper.
}

\hide{
\subsection{Data Persistence}
As mentioned in SubSection~\ref{subsec:motivation} there are two flavors of Bizur: an in-memory and
a persisted one. In accordance with allowing each IO to finish independently of others, Bizur persists
its state on a per bucket basis.
Bizur's persistent layer assumes good performance for random reads and writes; in our system we use SSDs.

There are three main on-disk structures:
\begin{enumerate}
    \item $head[]$ a list of pointers to buckets' content.
    \item The used blocks containing buckets.
    \item Free blocks, maintained as a free-list.
\end{enumerate}

When storing bucket $i$, Bizur first allocates free disk blocks needed for the bucket. It then serializes
the bucket's content to those blocks and updates $head[i]$ to point to the new blocks. Lastly, it frees the
old disk blocks which contained the bucket's old content, returning them to the free-list.

This double-update (first writing the content on free blocks, then updating the block-pointer)
ensures that Bizur's persistence is crash consistent, and that it will be able to recover consistently
from power failure.
}


\subsection{Optimizations}
When implementing the Bizur algorithm, we've employed multiple optimizations to improve the performance
and reduce the overhead. This section describes some of the optimizations.

The most straight-forward optimization is to avoid checking the leader is still the leader
(\lineref{line:leader_still_leader_1}-\lineref{line:leader_still_leader_2} in Algorithm~\ref{readbucket})
if \textsc{EnsureRecovery} indeed performed recovery. The recovery flow itself validates that the
leader still has a majority following it, and thus there is no need for an additional check.

A related optimization avoids sending the content of the bucket (from replicas to the leader)
when recovery is not needed. That is,
when \textsc{Read} needs to validate the leader still has a majority, there is no need to send the
content of the buckets to the leader (since the leader has the content locally); it is enough to
send \textsc{AckRead} without data.

The \textsc{Set} and \textsc{Delete} methods (see Algorithm~\ref{keyvalue}) first call
\textsc{Read}, then call \textsc{Write}. In normal operation (when recovery is not needed),
the \textsc{Read} path performs a cluster-wide validation that the leader has majority.
However, the \textsc{Write} will also perform this cluster-wide validation, as part of sending
the data to the replicas. To improve upon this,
if the bucket is already recovered, the leader skips the \textsc{Read} altogether, and uses its local
copy of the bucket instead. Recall that \textsc{Read}'s goal (when the bucket is already recovered)
is to guarantee the leader is still the leader. However, since replicas won't accept the write if
they have voted for a newer $elect\_id$, we get that guarantee during the write as well. This optimization
allows us to do just one cluster-wide message round-trip for \textsc{Set} and \textsc{Delete}, in
the common case of a bucket that has been recovered.

Regarding performing recovery on a bucket, we mentioned previously that it can be done lazily on the first
access to the bucket. However, the recovery process incurs additional latency, albeit minimal, and we
want to eliminate it when possible. To that end, when a leader gets elected, it will start a background process
that will slowly go over the entire data set, and recover all bucket. Thus ensuring that eventually
all buckets (even empty ones) have been recovered, and all future IOs
will not require recovery.

An optimization we haven't implemented, aims to reduce the load and time it takes for this background
process to run. The straightforward implementation recovers every bucket, as a separate operation. In some
cases, for example when there is almost no data, it is possible to batch the recovery of multiple buckets
into a single message round-trip.

A similar optimization, which we have implemented, applies to \textsc{IterateKeys}. Since \textsc{IterateKeys}
needs to go over all the buckets and extract the keys, it is very efficient to batch it. Moreover, since
\textsc{IterateKeys} just reads the buckets, we can do this batching at the leader. More precisely,
we need to check the leader is still the leader of the bucket range we're batching, then we can process
that entire range locally at the leader, without communicating with the other servers. This provides
are very efficient implementation of \textsc{IterateKeys}, both latency-wise and w.r.t. the load it
creates on the cluster.

\section{Scalability}\label{sec:scalability}
Bizur scales linearly, so long as the workload distributes more or less evenly across the key space.
Bizur's API lends itself nicely to sharding, since each range of buckets can be stored on a single
shard. Since the requirements around leadership are bucket-oriented (i.e., each bucket needs a single leader,
but different buckets can have different leaders), each shard can have its own leader; so long that
a given bucket is handled by a single shard.

Bizur can scale dynamically, by expanding and contracting. The scaling is shard based:
when the system first starts, all shards are on the same servers. As servers are added,
shards are migrated to those new servers. The shard migration
occurs using the same reconfiguration described in Section~\ref{sec:reconfig}. To shrink the system
back, the shards are migrated back, and the servers that don't have any shards on them can be removed.

Using a static number of shards (256 in our case) simplifies the code, yet allows it to grow to
very large clusters. It is possible to implement expanding/contracting by splitting and merging shards
to achieve a dynamic number of shards. We found the simplified version (static number of shards)
to be sufficient.

The \textsc{Get, Set} and \textsc{Delete} operations operate on a single key, and thus on a single bucket.
Hence, the more servers we have, the more operations we can do concurrently, keeping the same latency.
On the other hand,
the \textsc{IterateKeys} operation requires to go over the entire key set. The more servers we
have, the more work each \textsc{IterateKeys} needs to do. However, since the shards are independent,
the execution of a single \textsc{IterateKeys} operation can occur concurrently on all shards, thus
spreading the required work across the entire cluster.

\hide{
\begin{remark}
    Notice that if the workload doesn't distribute evenly across the shards, one can
    split the shards differently, according to the ``weight'' each bucket has. In our case, the
    workload does distribute evenly, so we haven't pursued this venue further.
\end{remark}
}

The ability to scale linearly is due to the less general API of Bizur.
This allows us to avoid reimplementing scaling
flows over and over again, in each application that uses Bizur.

\section{Testing}
To test the strong consistency (strict serializability) of the Bizur implementation, we've developed a testing
tool called \textit{Serialla}\footnote{Serialla is a generic tool, used to test strict serializability
of additional services in Elastifile's file-system, not just Bizur.}.

Serialla is a randomized testing
tool, that executes concurrent operations against the Bizur, while tracking the responses.
It then looks for a strict serializable execution that can explain the operations and responses.
If it finds such an execution, it continues testing with another batch of concurrent operations.
If no such execution exists, Serialla will report the problem, together with a descriptive log
of the concurrent operations; passing the responsibility on to a human developer to find and fix the bug.

Serialla tries to create as much \textit{chaos} as possible: developers can annotate areas of the code
which are sensitive to races. Serialla will then (randomly) try to explore these races, by scheduling
in / out the relevant threads. For example, when taking a lock, usually the lock will not be contended,
and regular tests will not explore the races around it. By using the annotation, Serialla will explore
different possible schedules (i.e., races) around the lock, even when it is not contented, thus
flushing out rare bugs.

\section{Experimental Evaluation}\label{sec:experimental}
The cluster consists of 4 Supermicro servers,
equipped with an Intel Xeon E5-2620 processor running at 2.40GHz. Each server
has a 10-Gigabit Intel 82599EB NIC, and a 128GB SanDisk SATA SSD device. All servers
run Linux CentOS 7.1 within a VM (NICs are configured with SR-IOV and SSDs are configured with passthrough).

Our tests utilize a single core on each machine, where 3 servers run the consensus algorithm,
and one server creates the load. All tests were done with 100\% writes (set operation) of small
values, 50 bytes each. We avoided
reads (get operations) since the current version of etcd performs reads directly from the leader,
without contacting the cluster, which doesn't preserve the same consistency level as Bizur and ZooKeeper do.

Three different key-value systems are evaluated: etcd v3.1.0-rc.1 (which uses the Raft consensus algorithm),
ZooKeeper 3.4.9 (which uses the Zab consensus algorithm) and Bizur\footnote{We
benchmark the persistent flavor of Bizur, to compare
    apples to apples (see Subsection~\ref{subsec:motivation}).}. etcd and ZooKeeper were evaluated
with their default configuration parameters.

\hide{
\begin{remark}
    As mentioned in Subsection~\ref{subsec:motivation}, we have two different usages of Bizur. One that is
    in-memory only, and another one that is persisted. In this entire
    section we benchmark only the persisted flavor of Bizur, to ensure our evaluation compares
    apples to apples.
\end{remark}
}

\hide{\begin{remark}}
    Section~\ref{sec:scalability} describes the ability of Bizur to scale by utilizing multiple shards.
    Since etcd and ZooKeeper do not support this, the following benchmarks contain a single shard only.
    Notice that Bizur can reach much higher throughput, while preserving the same low-latency, if it is allowed
    to use multiple shards.
\hide{\end{remark}}

The different systems
are written in different languages: Bizur is written in C, etcd is written in Go and ZooKeeper is written
in Java. The difference in programming language could affect the performance of the system, but
isn't expected to affect the behavior of a system. That is, etcd's lower throughput can be explained by
Go being slower than C; but etcd's throughput drop during leader failure is hard to attribute to the Go language,
and is most likely related to the Raft protocol itself.

First, we evaluate the effect on performance, of the concurrency and of the number of keys in the data
set (see Subsection~\ref{subsec:conandkeyset}). Second, we evaluate the effect of packet drops
on the performance of the different systems (see Subsection~\ref{subsec:packetdrop}).
Lastly, we show the behavior of the different systems when the leader fails
(see Subsection~\ref{subsec:leaderfailure}).

\subsection{Concurrency and Key-set Size}\label{subsec:conandkeyset}
Bizur's high concurrency stems from the independence of operations. Thus, to fully utilize Bizur,
there must be enough keys so that concurrent operations work on different keys. Figure~\ref{fig:keys}
shows the throughput, average latency and 99\textit{th} percentile latency for different number of
keys. The concurrency (queue depth) of operations is fixed to 64 throughout this benchmark.

For the Bizur system, once the key-set size reaches the concurrency (queue depth = 64), the average latency's
improvement - as well as the throughput's improvement - flattens out. This is expected, since once
each concurrent operation has a different key, Bizur is at it's optimal performance. The 99\textit{th} percentile
latency's improvement requires more keys to flatten out; this is also expected, since this percentile
measures boundary latencies, which happen when concurrency operations share the same key.
When the key-set size is close to the queue depth, it is still likely that there will be
concurrent operations sharing the same key.

Both ZooKeeper and etcd have a flat line w.r.t. throughput and average latency. This is expected,
as both algorithms serialize all operations, so the number of keys shouldn't affect the performance.
However, both algorithms have a somewhat odd behavior of the 99th percentile with larger number of keys:
the 99th percentile latencies decrease as the key-set size increases. We do not have a good explanation
for this behavior.

Figure~\ref{fig:qd} compares the performance as the concurrency increases. All
three systems present the same (expected) behavior:
the latency increases as the queue depth increases, while the throughput's improvement
flattens out at some point. Notice the Bizur's 99th percentile latency increases noticeably slower
than that of etcd and ZooKeeper, and are much closer to the average, meaning that the Bizur's latency
variance is much tighter than etcd's or ZooKeeper's.

We've chosen queue depth of 64 as a good trade-off
between throughput gain and latency cost. All following benchmarks are done with queue depth
of 64.

\begin{figure*}[!h]
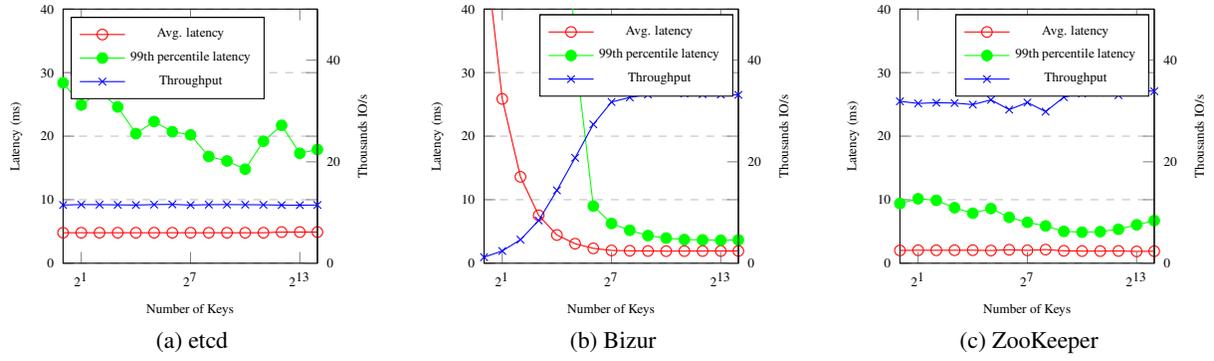

\begin{subfigure}{.335\textwidth}
  \centering
  \includegraphics[width=.9\linewidth]{etcd_iops_by_keys}
  \vspace{-5pt}
  \caption{etcd}
\end{subfigure}
\begin{subfigure}{.335\textwidth}
  \centering
  \includegraphics[width=.9\linewidth]{bizur_iops_by_keys}
  \vspace{-5pt}
  \caption{Bizur}
\end{subfigure}%
\begin{subfigure}{.335\textwidth}
  \centering
  \includegraphics[width=.9\linewidth]{zookeeper_iops_by_keys}
  \vspace{-5pt}
  \caption{ZooKeeper}
\end{subfigure}
\vspace{-8pt}
\caption{Effect of number of keys (queue depth = 64)}
\label{fig:keys}
\end{figure*}
\vspace{5pt}

\begin{figure*}[!h]
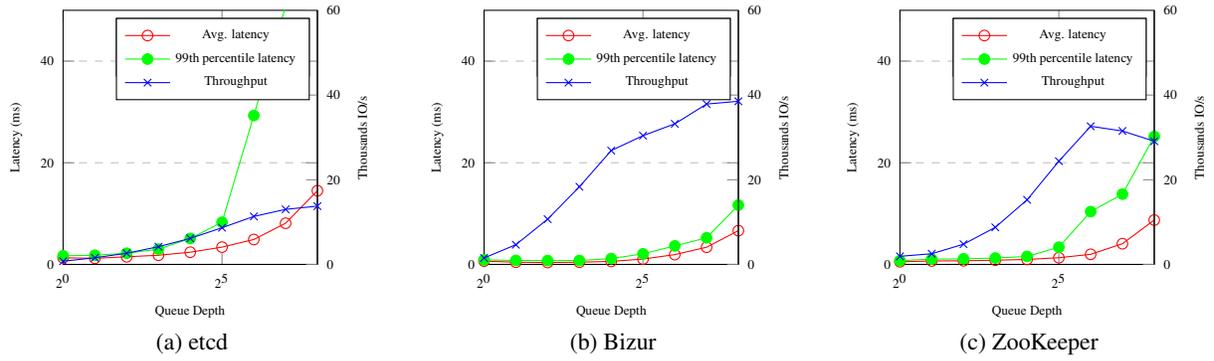

\begin{subfigure}{.335\textwidth}
  \centering
  \includegraphics[width=.9\linewidth]{etcd_iops_by_qd}
  \vspace{-5pt}
  \caption{etcd}
\end{subfigure}
\begin{subfigure}{.335\textwidth}
  \centering
  \includegraphics[width=.9\linewidth]{bizur_iops_by_qd}
  \vspace{-5pt}
  \caption{Bizur}
\end{subfigure}%
\begin{subfigure}{.335\textwidth}
  \centering
  \includegraphics[width=.9\linewidth]{zookeeper_iops_by_qd}
  \vspace{-5pt}
  \caption{ZooKeeper}
\end{subfigure}
\vspace{-8pt}
\caption{Effect of queue depth (number of keys = 16,000)}
\label{fig:qd}
\end{figure*}
\vspace{5pt}

\begin{figure*}[!h]
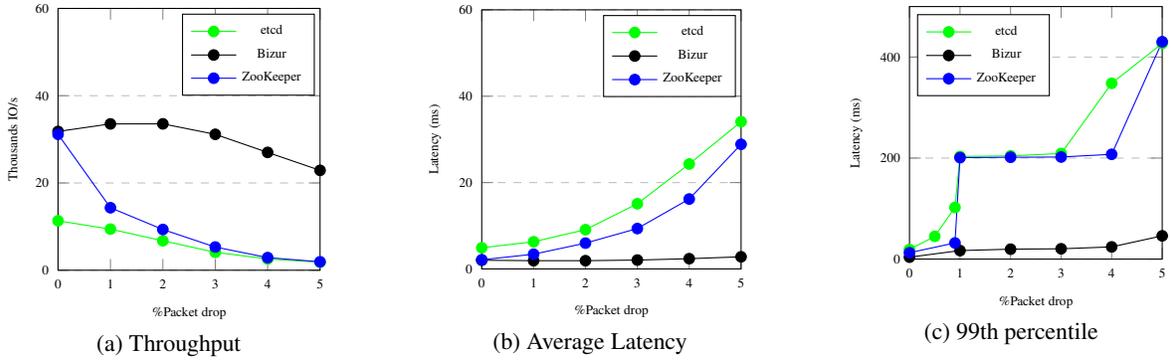

\begin{subfigure}{.335\textwidth}
  \centering
  \includegraphics[width=.8\linewidth]{packet_loss_throughput}
  \vspace{-5pt}
  \caption{Throughput}
\end{subfigure}
\begin{subfigure}{.335\textwidth}
  \centering
  \includegraphics[width=.8\linewidth]{packet_loss_latency}
  \vspace{-5pt}
  \caption{Average Latency}
\end{subfigure}
\begin{subfigure}{.335\textwidth}
  \centering
  \includegraphics[width=.8\linewidth]{packet_loss_99_percentile}
  \vspace{-5pt}
  \caption{99th percentile}
\end{subfigure}
\vspace{-8pt}
\caption{Effect of packet drop}
\label{fig:packet_loss}
\end{figure*}
\vspace{5pt}

\begin{figure*}[!h]
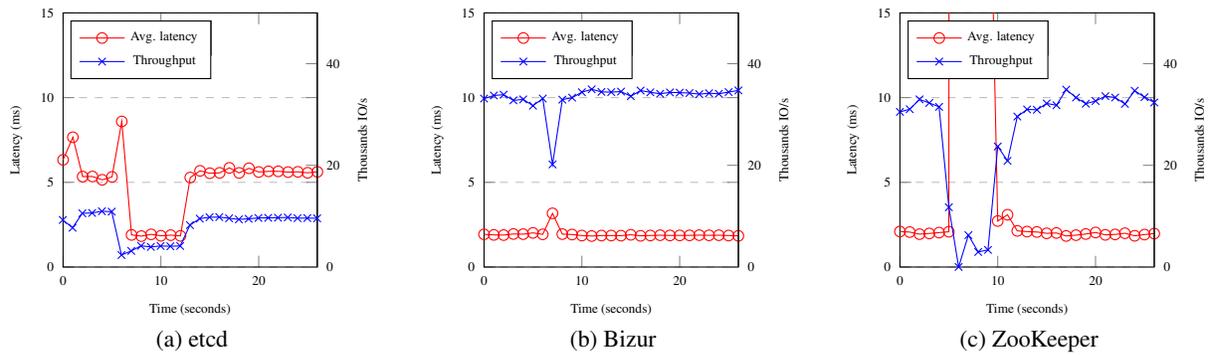

\begin{subfigure}{.335\textwidth}
  \centering
  \includegraphics[width=.9\linewidth]{etcd_leader_change}
  \vspace{-5pt}
  \caption{etcd}
\end{subfigure}
\begin{subfigure}{.335\textwidth}
  \centering
  \includegraphics[width=.9\linewidth]{bizur_leader_change}
  \vspace{-5pt}
  \caption{Bizur}
\end{subfigure}%
\begin{subfigure}{.335\textwidth}
  \centering
  \includegraphics[width=.9\linewidth]{zookeeper_leader_change}
  \vspace{-5pt}
  \caption{ZooKeeper}
\end{subfigure}
\vspace{-8pt}
\caption{Effect of Leader Failure}
\label{fig:leader_change}
\end{figure*}

\subsection{Packet Drops}\label{subsec:packetdrop}
In this benchmark we compare the behavior of the systems under packet drop, ranging from 0\% drop
and up to 5\% packet drop. Packet drops are expected to have little effect on Bizur, due to independence
of operations. Thus, if one operations slows down due to a network issue, other operations just continue
running. This expected behavior can be seen in Figure~\ref{fig:packet_loss}, where the average
latency is increased slightly, and the throughput decreases by about 25\%.

ZooKeeper and etcd, on the other hand, have a much more noticeable performance drop. ZooKeepers'
average latency goes up from 2.1 milliseconds to 28 milliseconds, more than x10. etcd's
average latency goes up from 4.9 milliseconds to 34, a x7 increase.
Throughput of ZooKeeper degrades from 31K to 2K, a factor of x15, while etcd degrades from 11K to 2K,
a factor of x5.

Figure~\ref{fig:packet_loss}.c displays the 99\textit{th} percentile.
Both ZooKeeper and etcd have an odd behavior: there are two sharp jumps in latency
around 1\% packet drop, then a again around 3-4\%, eventually reaching more than
400 millisecond latency. It is not clear what causes this behavior, we postulate it might be an
artifact of the TCP implementation in the kernel.

Bizur's 99\textit{th} percentile latency behaves much better:
increasing steadily up to about 45 milliseconds for 5\% packet drop.
Bizur uses UDP for its networking, which is aligned with the previous assumption about the TCP implementation.

The packet drop test shows the advantage Bizur has over distributed log based consensus algorithms, like
Raft and Zab. The false and real dependencies incurred by the distributed log are clearly evident in
the behavior during packet drop. Bizur, on the other hand, behaves very well as the packet drop increases.

\subsection{Leader Failure}\label{subsec:leaderfailure}
In this benchmark we compare the behavior of the systems when the leader fails.
We measure 5 seconds of performance prior to killing the leader, and another 20 seconds
following the leader's failure.

As can be seen in Figure~\ref{fig:leader_change}, Bizur has a short drop in throughput, then
goes back to normal operation. ZooKeeper takes about 7 seconds to go back to normal
operation, with large drops in throughput and huge latencies following the failure. ZooKeeper
doesn't provide any IOs for the first second after the failure.

etcd's throughput goes down by about 60\%, and takes 7 seconds to go back up to its pre-failure
values. etcd has a peak in latencies immediately following the failure, after which the latencies
go down to levels lower than before the failure. As the throughput returns to normal, so do the
latencies. Interesting to notice that the lower latencies etcd achieves are the normal latencies
ZooKeeper and Bizur have.

This test emphasizes the benefit of the Bizur algorithm: due to its built-in concurrency design,
it has low cost of leader-change, and can thus
support very short failure detection timeouts. That's why it can return to normal operations so
quickly.

\section{Related Work}
Consensus algorithms have attracted a lot of attention in areas requiring high-availability
and fault tolerance. Mainly because they provide the ability to replicate a service across multiple
servers, in a way that allows the service to continue working when some of the servers fail.
This scheme of using consensus to achieve a replicated service is called ``state machine replication'' (SMR)
\cite{schneider1990implementing}.

Multiple algorithms fall in the SMR category:
the original Paxos \cite{lamport1998part}, ZooKeeper's \cite{hunt2010zookeeper} consensus algorithm
Zab \cite{junqueira2011zab}, Raft \cite{ongaro2014search} which aims to simplify Paxos,
and Viewstamp replication \cite{liskov2012viewstamped}.

Among the above, Paxos got the most attention, and is used in large-scale
production environments \cite{burrows2006chubby,chandra2007paxos,corbett2013spanner} as well
as by many distributed file-systems \cite{nutanixbible, cephusingpaxos, xtreemfsusingpaxos, infinituseingpaxos}.

Many research papers have improved upon the original Paxos. For example, Fast paxos \cite{lamport2006fast}
which adds ``fast rounds'' that reduces the number of message delays required to learn a value,
or Mencius \cite{mao2008mencius} which round-robins the proposing server, to evenly spread the leadership
load. Such improvements still suffer from both the false and the real dependencies (see Subsection~\ref{subsection:distlog}).

Others have improved Paxos to be more concurrent, by removing the false dependencies between log entries.
For example, Generalized Paxos \cite{generalized-consensus-and-paxos} allows concurrent log entries
as long as they commute. Multicoordinated Paxos \cite{camargos2007multicoordinated} and
Egalitarian Paxos \cite{moraru2013there} achieve similar improvements.

The above still suffer from real dependencies between log entries, and as such have the drawbacks
mentions in Subsection~\ref{subsection:distlog}. Algorithms that rely on distributed log in some form
or another will not be able to get rid of real dependencies, as they're essential to the manner
in which the state is replicated.

Others have taken atomic memory as their building block. Examples include RAMBO \cite{gilbert2010rambo}
and a similar work \cite{aguilera2011dynamic} showing reconfiguration without consensus. Both support
\textit{read} and \textit{write} operations, but do not support \textit{conditional write}, which
is an important primitive for distributed services. In addition, they concentrate on individual objects
and do not handle having multiple objects and addressing them.

DO-RAMBO \cite{georgiou2009developing} handles multiple objects. However, it still doesn't support
\textit{conditional write}, and it assumes a fixed namespace of the multiple objects it handles.
Thus, a general key-value cannot be used directly, since each key must be known as part of the fixed namespace.

Another system providing consistent key-value like behavior appears in \cite{van2004chain},
which doesn't support \textit{conditional write} but it seems that it can be added easily.
However, \cite{van2004chain} assumes FIFO links, which incur false dependencies between operations.
It might be possible to remove this assumption, but it would probably make the failure handling much
more complex.

Another area, which is sometimes mixed with the above consensus-related field, is that of
object-store systems like Cassandra \cite{lakshman2010cassandra} or MongoDB \cite{banker2011mongodb}.
These systems were designed to provide availability over consistency w.r.t.
the \textit{CAP theorem} \cite{gilbert2002brewer}.
As such they originally were not strongly consistent. There are conflicting claims regarding their
consistency today. Specifically, claims of strong consistency have been refuted (see \cite{aphyrcassandra,aphyrmongodb}).\footnote{
    ZooKeeper, on the other hand, has passed the strong consistency tests \cite{aphyrzookeeper}.
}

Lastly, high-throughput systems like \cite{199315,dragojevic2015no}
assume an external configuration service (i.e., Paxos),
while systems like \cite{zhang2013transaction,lloyd2011don} add additional assumptions to the consistency model.

\section{Conclusion}
State machine replication (SMR) and Paxos-like consensus algorithms are used in many distributed
systems in general, and in scale-out file-systems in particular. The general data-model
supported by such algorithms imposes performance constraints: limited scalability and concurrency,
as well as latency issues during failures or packet drops.

However, when looking closely at the requirements of distributed systems, and especially at
distributed file-systems, it is possible to weaken the data model to a key-value data model.
Bizur is a consensus algorithm that exposes a key-value API, and overcomes the scalability
limits of Paxos-like algorithms, as well as guarantees low-latency even during failures.

In some systems, SMR is used as an underlying infrastructure of a key-value interface (e.g., etcd).
In such places, replacing the key-value layer - that rides on top of a distributed log consensus algorithm -
with Bizur doesn't even weaken the data model. It just gains performance.

Through comparative benchmarking with ZooKeeper and etcd, we've shown that Bizur outperforms them
in throughput, average
latency and 99th percentile latency; both in regular operation, during network packet drops
and especially during server failure.

\section{Acknowledgments}
We would like to thank Assaf Yaari, Erez Yaffe, Eli Weissbrem, Amir Levy, Nadav Shemer,
Shai Koffman and Shahar Frank for their feedback and help.

{\footnotesize \bibliographystyle{acm}
\bibliography{bizur}}

\end{document}